\documentclass[conference]{IEEEtran}
\usepackage{graphicx}
\usepackage{amsmath,amssymb,epsfig,xcolor} 
\usepackage{subfig}
\usepackage{tabularx,multirow,array}
\usepackage{times}
\usepackage{booktabs}
\usepackage{multirow}
\usepackage{soul}
\usepackage{enumerate}
\usepackage{xcolor}
\usepackage[countmax]{subfloat}
\usepackage{algpseudocode}
\usepackage{algorithm}
\usepackage{threeparttable}
\usepackage{cite}
\usepackage{url}
\usepackage{amsthm}

\theoremstyle{plain}
 \newtheorem{thm}{Theorem}

\SetMathAlphabet{\mathcal}{bold}{OMS}{cmsy}{b}{n}
\begin{document}
\title{Modeling of Deep Neural Network (DNN) Placement and Inference in Edge Computing}
\author{
    \IEEEauthorblockN{Mounir Bensalem, Jasenka Dizdarevi{\'{c}}  and Admela~Jukan}
    \IEEEauthorblockA{Technische Universit\"at Braunschweig, Germany}
     \IEEEauthorblockA{\{mounir.bensalem, j.dizdarevic, a.jukan\}@tu-bs.de}
}
\maketitle
\begin{abstract}

With the edge computing becoming an increasingly adopted concept in system architectures, it is expected its utilization will be additionally heightened when combined with deep learning (DL) techniques. The idea behind integrating demanding processing algorithms in Internet of Things (IoT) and edge devices, such as Deep Neural Network (DNN), has in large measure benefited from the development of edge computing hardware, as well as from adapting the algorithms for use in  resource constrained IoT devices. Surprisingly, there are no models yet to optimally place and use machine learning in edge computing. In this paper, we propose the first model of optimal placement of Deep Neural Network (DNN) Placement and Inference in edge computing. We present a mathematical formulation to the DNN Model Variant Selection and Placement (MVSP) problem considering the inference latency of different model-variants, communication latency between nodes, and utilization cost of edge computing nodes. We evaluate our model numerically, and show that for low load increasing model co-location decreases the average latency by 33\% of millisecond-scale per request, and for high load, by  21\%.
\end{abstract}

\section{Introduction}
The potential benefits of edge computing paradigm and related distributed system solutions, have been particularly linked with the breakthroughs achieved in the fast growing development of deep learning (DL) techniques designed to boost automation in all application domains. With that in mind, this vibrant research area has been more and more focusing on integrating edge computing with deep learning \cite{Han2019} and the associated challenges due to resource constraints \cite{wu2019machine, Fowers2018}.  
Recent hardware developments are making more and more possible to run highly computationally demanding algorithms in the edge \cite{Ramneek}. 

Among myriad of open research issues,  the models for machine learning (ML) inference latency and ML model selection optimization in edge computing, along with related task placement are of particular importance. This is because the related  such models \cite{romero2019infaas} developed for cloud computing cannot be directly applied in edge computing. The DNN placement problem in the edge needs to consider in particular the communication delay between nodes and the hardware heterogeneity of devices. To the best of our knowledge there has been no study of the DNN application selection, placement and inference serving problem in consideration of edge computing. This paper presents the first DNN Model Variant Selection and Placement (MVSP) in edge computing networks. We provide a mathematical formulation of the problems of ML placement and inference service, considering inference latency of different model-variants,  communication latency between nodes and utilization cost of edge computing nodes (resources). Our model also includes a discussion on the potential effects of hardware sharing, with GPU edge computing nodes shared between different model-variants, on inference latency. 


An illustration of the DNN application placement problem is presented in Figure \ref{fig:arch1} with the arrival of inference requests from IoT nodes to the edge computing layer.  IoT nodes are assumed to be devices with processing and sensing capabilities, but not enough to run DNN models. In this system abstraction, edge computing layer, consisting of edge nodes with GPUs for running ML models,  serves as an inference service system to the requests from IoT nodes. For the illustrated system we focus on designing a placement strategy of ML models, taking into account different possibilities of model-variants and how to forward requests coming from IoT nodes.

%

\begin{figure}
 \centering
   \includegraphics[scale=0.45]{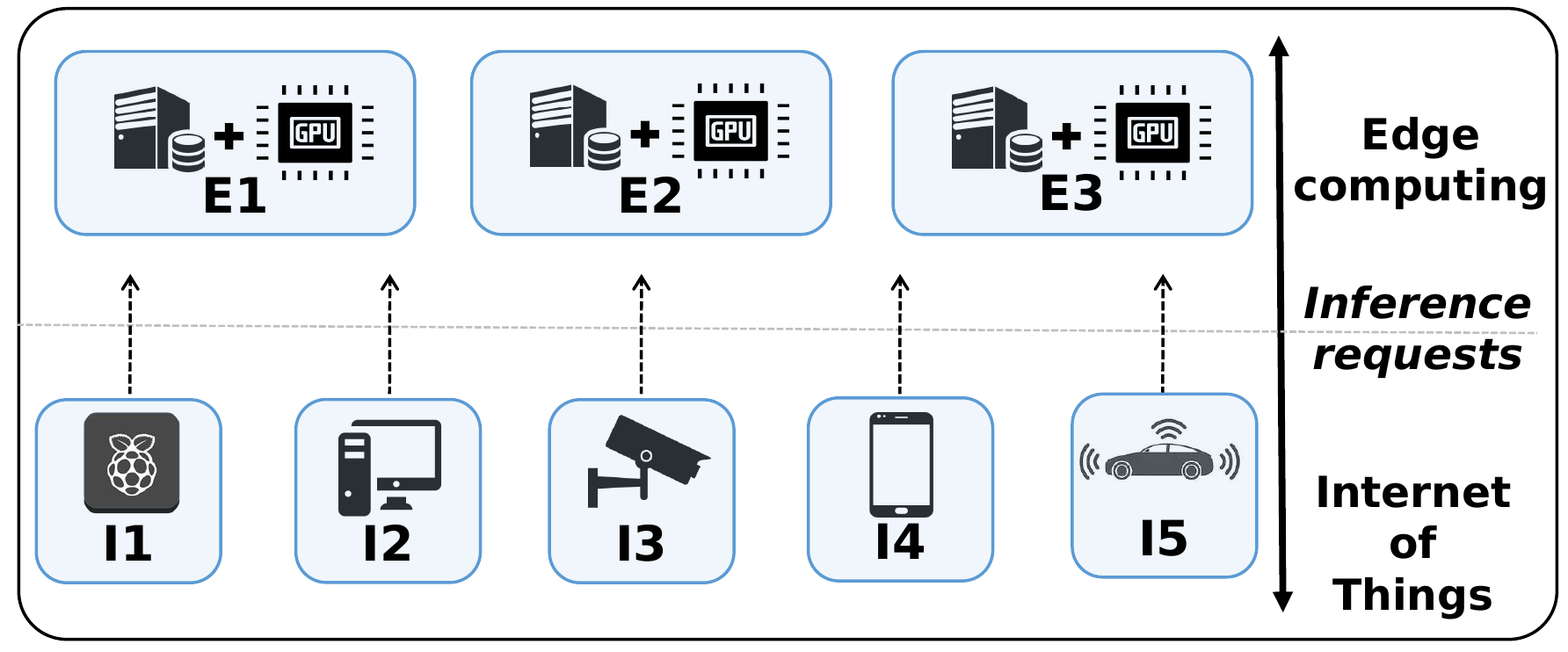}
 \caption{A reference inference service system}
\label{fig:arch1}
\vspace{-0.2 cm}
\end{figure}
\begin{figure*}
 \centering
   \includegraphics[scale=0.8]{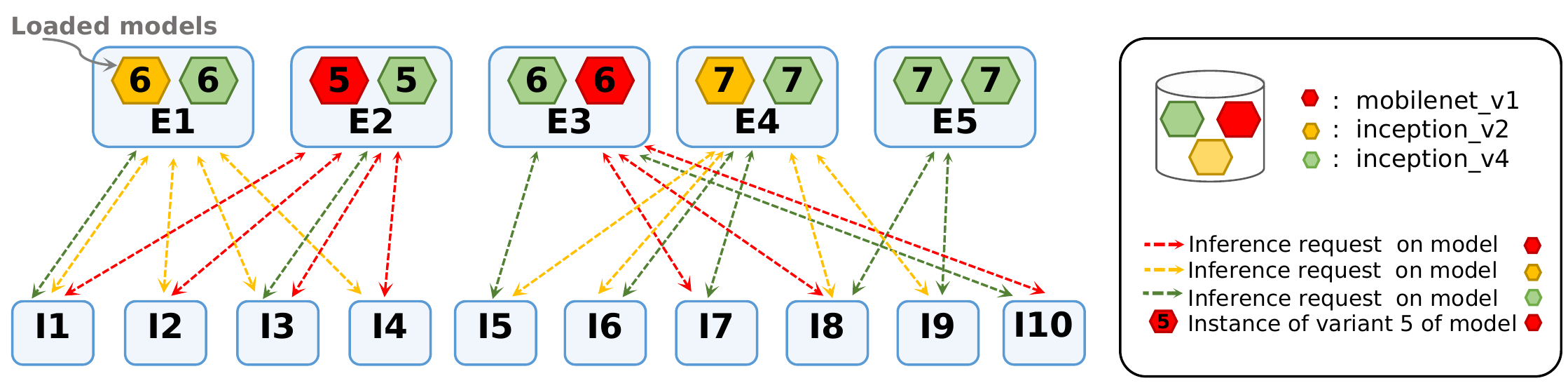}
 \caption{An illustrative example of DNN application placement problem and assignment of inference request in edge nodes}
\label{fig:arch2}
\vspace{-0.2 cm}
\end{figure*}

The remainder of this paper is organized as follows. Section \ref{sec:SM} introduces a mathematical model for the MVSP problem.  Section \ref{sec:results} numerically evaluates the proposed model. We conclude the paper in \ref{sec:conclusion}.

\section{System Model}\label{sec:SM}
\subsection{Reference Edge Computing Network Model}
In order to analyze MVSP problem in edge computing network we define a system model which will consider inference latency of different model-variant with shared and unshared access to GPUs, node communication latency and utilization cost. 
The considered system consists of $N_I$ IoT nodes, e.g. smart phone, security camera and smart car cameras, and  $N_E$ edge nodes, e.g. access points. Let $\mathcal{N_I}=\{1,...,N_I\}$ and $\mathcal{N_E}=\{1,...,N_E\}$  denote the set of indexes of IoT nodes and edge nodes, respectively. Edge nodes are able to host various ML applications designated to serve the inference requests coming from IoT nodes. Every edge node $e\in \mathcal{N_E}$ has a computing unit specific for inference serving tasks e.g. CPU, GPU and TPU as well as memory capacity $C^{e}$.  We assume that we have $M$ different ML models that can be used for different tasks such as face recognition and object detection. Each ML model $m$ can have $V_m$ variants with different sizes and inference latencies per request and can be deployed via a VM or a container. We denote  by $\mathcal{M}=\{1,...,M\}$ the set of ML models and  $\mathcal{V}_m=\{1,...,V_m\}$ the set of variants of model $m$. Each model variant ($m,v$) has a minimum memory requirement  $R_{mv}$ to be loaded and can process at most $Load_{mv}$ with a stable performance. Each IoT node $i$ can define its own latency requirement $Lmax_{m}^{i}$ for each infered model $m$ as well as the number of inference requests $r_m^i$.  The notations used in this paper are summarized in Table \ref{tab:TableOfNotationForMyResearch}.
We introduce a binary variable $x_{mv}^{ie}$ to indicate the forwarding decision of requests of model-variant ($m,v$) from IoT node $i\in \mathcal{N_I} $ to edge node $e\in \mathcal{N_E}$. The placement decision of model-variant ($m,v$) in an edge node $e$ is defined by an integer variable $n_{mv}^{e}$, which indicates the number of deployed instances. \\
Figure \ref{fig:arch2} shows an illustrative example of a network of $10$ IoT nodes ($I1,...,I10$) and $5$ edge nodes ($E1,...,E5$) for the above described system. In this example, each edge node stores the 3 ML models and can instantiate them during loading various model-variants by changing the batch size parameter, which affects the instance size and throughput. The figure shows how from this set of 3 models, the optimally selected model-variants would be placed in edge nodes after the placement decisions have been made, along with served inference requests. For example, after placement edge node E1 has two loaded models, \emph{inception\_v2} and \emph{inception\_v4} and five served inference requests, four for \emph{inception\_v2} and one for \emph{inception\_v4}. Serving of inference requests is achieved by assigning them to the appropriate edge node by considering the latency requirements, capacity and the cost of using servers. We assume that each IoT node consumes these 3 different ML models with different request rates.

\subsection{Latency Model}
We consider two types of latencies: communication latency between IoT and edge nodes, and inference latency of model-variants in edge nodes. We denote by $CL^{ie}$ the communication latency between IoT node $i$ and edge node $e$. The inference latency of a model-variant ($m,v$) running exclusively on edge node $e$ is denoted by $L_{mv}^{e}$. We define $IL_{mv}^{e}$ as the inference latency of a model-variant ($m,v$) running on edge node $e$. For mathematical model of this latency we include the effects that sharing with other model-variants can have, as well as observing the case with unshared access to GPU. With that in mind, we assume in our formulation that an edge node can be shared by at most $K$ model-variants. The average latency per request is given by:

\begin{equation}\label{eq:avglatency}
\begin{split}
L = & \frac{1}{\sum_{i\in \mathcal{N_I}} \sum_{m\in \mathcal{M}} r_{m}^{i}}\\& \sum_{i\in \mathcal{N_I}} \sum_{e\in \mathcal{N_E}} \sum_{m\in \mathcal{M}} \sum_{v\in \mathcal{V}_m}  r_{m}^{i}(CL_{i}^{e} + IL_{mv}^{e}  ) x_{mv}^{ie}
\end{split}
\end{equation}
\\
The communication latency between node $i\in \mathcal{N_I}$ and  node $e\in \mathcal{N_E}$ is considered as the sum of the delay on each link in the shortest path in both directions (sending request and receiving response). The delay $d_l$ on each link $l$ is assumed to have a random value with an average $\mu$, including all possible existing delays in the link i.e. transmission, queuing, propoagation and processing. We denote by $P_{ie}$ the set of links in the shortest path between node $i$ and $e$.\\
\begin{equation}
CL_{i}^{e}=\sum_{l\in P_{ie}}d_{l}
\end{equation}

We model the inference latency of a model-variant ($m,v$) running on edge node $e\in \mathcal{N_E}$, such that: the inference latency of a model-variant increases linearly in terms of the latency of co-located model-variants. A discussion on resource sharing is provided in subsection \ref{sec:DRS}. The expression of inference latency is given by:\\

\begin{equation}
\begin{split}
IL_{mv}^{e} = & L_{mv}^{e}+ \alpha_{mv} L_{mv}^{e} (n_{mv}^{e} - 1) \\ & + \sum_{m'\in M} \sum_{\substack{v'\in \mathcal{V}_m \\ v'\neq v}} \alpha_{m'v'} L_{m'v'}^{e} n_{m'v'}^{e}
\end{split}
\end{equation}

where the inference latency is the sum of the inference latency of a model-variant running exclusively on an edge node ($L_{mv}^{e}$), the additional latency created by replication and the additional latency created by co-locating a different model on the same node.\\

IoT nodes (users) are assumed to express their latency requirements for the inference of a model $m\in \mathcal{M}$ with a latency requirement constraint  given by:
\begin{equation}\label{cons:latencyreq}\begin{split}
 \sum_{v\in \mathcal{V}_m}  (CL_{i}^{e} + IL_{mv}^{e}) x_{mv}^{ie}  \leq Lmax_{m}^{i}\\ \;\; \forall i \in \mathcal{N_I}, \;\forall e \in \mathcal{N_E}, \; \forall m \in \mathcal{M}
 \end{split}
\end{equation}
This constraint assumes that the round trip time (RTT) cannot exceed a maximum value of latency given by the user as a requirement $Lmax_{m}^{i}$. In our case, RTT is the sum of the communication delay (cumulative delay among the path) and the processing delay of inference request in the edge node.
\subsection{Utilization Cost Model}
The utilization cost model is an abstract formulation of all the costs induced from the utilization of edge resources, assuming that such a cost would increase with the increase of resource utilization. As an example, the utilization cost can represent the power consumption and energy efficiency measurements in a unity of power (Watt), considering different hardware components such as CPU, GPU, memory, and I/O. \cite{JetsonAGXbenchmarks} shows some inference benchmarks of several DNN model-variant using Jetson AGX Xavier, Nvidia GPU. 

For the sake of generality, we define a continuous variable $z^e$ denoting the utilization cost of a node $e\in \mathcal{N_E}$. The average utilization cost of all edge nodes is given as:
 
\begin{equation}\label{eq:avgcost}
C = \frac{1}{N_E} \sum_{e\in \mathcal{N_E}}z^e
\end{equation}

The cost of every edge node $e$ is related to its memory utilization $u_e$.  Similarly to \cite{carpio2018balancing}, the utilization cost follows an exponential function $\phi(.)$ of the utilization. We denote by $\mathcal{Y}$ a  set of linear functions tangent to $\phi(.)$. Using the set of linear functions $\mathcal{Y}$, we approximate the utilization cost as follows:
\begin{equation}
\phi(x) = \max_{y\in \mathcal{Y}}y(x), \forall x\in \mathbb{R}
\end{equation} 
The following constraints insure that the variable $z^e$ gets a value approximately equal to $\phi(u_e)$.
\begin{equation}\label{const:cost}
z^{e} \geq y (u_e) , \forall e \in \mathcal{N_E}, \; \forall y \in \mathcal{Y}
\end{equation}

\begin{equation}\label{const:util}
u_e = \frac{1}{C^e}\sum_{m\in \mathcal{M}}\sum_{v\in \mathcal{V}_m}R_{mv}n_{mv}^{e} \leq 1 ,\; \forall e \in \mathcal{N_E}
\end{equation}
Constraint (\ref{const:util}) shows the definition of an edge node utilization $u_e$, as the sum of utilization of all possible model-variants $mv$ i.e. in terms of required memory per loaded model. Then we divide the obtained sum by the memory capacity of a node (mainly GPU memory). 

\subsection{Loading and Scaling Model}

A model-variant can be loaded in a specific node to serve requests coming from users (IoT nodes). When the requests load increases the deployed model-variant may not be able to serve users. Considering this scenario, the system may replicate the model instance on top of a new VM or container which scale up the throughput  (based on container technologies) or use a different model-variant with a bigger batch size (i.e has higher fps).\\
The following constraints can insure that the load on a specific model variant on a specific node cannot exceed the maximum load $Load_{mv}$. Moreover, by conserving the maximum load, this constraint can scale-up the number of model-variant replicas (called \textit{variant replication} in \cite{romero2019infaas}), or choose to use another model variant that has less inference latency (called \textit{variant upgrading} in \cite{romero2019infaas}) regarding the minimization of the average latency (eq. \ref{eq:avglatency}). 
\begin{equation}\label{const:loading}
\sum_{i\in \mathcal{N_I}}  r_{m}^{i}x_{mv}^{ie} \leq Load_{mv} n_{mv}^{e} ,\;\; \forall e \in \mathcal{N_E}, \forall m \in \mathcal{M}, \forall v \in  \mathcal{V}_{m}
\end{equation}
The following constraints assure that the number of model-variant instances is bigger than $1$ only if at least one node is sending inference requests.
\begin{equation}\label{const:repl}
n_{mv}^{e} \geq x_{mv}^{ie} \;\; \forall i \in \mathcal{N_I}, \forall e \in \mathcal{N_E}, \forall m \in \mathcal{M}, \forall v \in  \mathcal{V}_{m}
\end{equation}
The number of instances shared by a node can be at most $K$:
\begin{equation}\label{const:repl2}
n_{mv}^{e} \leq K \;\;  \forall e \in \mathcal{N_E}, \forall m \in \mathcal{M}, \forall v \in  \mathcal{V}_{m}
\end{equation}
The memory capacity constraints can be defined as follows:
\begin{equation}\label{const:memory}
\sum_{m\in \mathcal{M}} \sum_{v\in \mathcal{V}_m} R_{mv}n_{mv}^{e} \leq C^{e},\;\; \forall e \in \mathcal{N_E}
\end{equation}
\subsection{Problem Formulation}
The Model Variant Selection and Placement (MVSP) problem aims to minimize both the average latency per request (eq. \ref{eq:avglatency}) and the average utilization cost of edge nodes (eq. \ref{eq:avgcost}):
\begin{equation}\label{obj:1}\begin{split}
   \textbf{P : }  & \min \alpha L + (1-\alpha)C \\
s.t. & \sum_{v\in \mathcal{V}_{m}}\sum_{e\in \mathcal{N_E}}x_{mv}^{ie}=1,\;\;\forall i\in \mathcal{N_I},\;\; \forall m\in \mathcal{M}\\
     & \text{Constraints: } (\ref{cons:latencyreq}),(\ref{const:cost}),(\ref{const:util}),(\ref{const:loading}), (\ref{const:repl}), (\ref{const:repl2}),(\ref{const:memory})
\end{split}
\end{equation}
where $\alpha$ denotes weight of the average latency in the objective function.
The first constraint in the problem (\ref{obj:1}) insures that a request from a specific IoT node can be processed only by one edge node. Constraint (\ref{cons:latencyreq}) assures that RTT cannot exceed the maximum tolerated latency. Constraint (\ref{const:cost}) and (\ref{const:util}) are used to compute utilization cost per edge node. Constraint (\ref{const:loading}) insures that the load assigned to a specific model-variant deployed in a specific edge node does not exceed its maximum processing capacity. Constraints (\ref{const:repl}) and (\ref{const:repl2}) defines the values of binary and integer variables used in the model. Constraint (\ref{const:memory}) assures the satisfaction of memory capacity per edge node.
\subsection{Complexity Analysis}
\begin{thm}
MVSP problem is NP-hard.
\end{thm}
\begin{proof}
MVSP is a mixed integer program with quadratic terms in the objective and in the constraints which is complex to solve. The quadratic terms can be linearized using \textit{standard linearization techniques} presented in \cite{glover1974converting} to obtain a solvable MILP. MVSP is NP-hard because it combines two NP-hard problems which are the model-variant allocation  problem  and the inference assignment problem. The model-variant allocation  problem can be obtained by a model relaxation which minimizes the cost under the capacity constraint (\ref{const:memory}). The problem is equivalent  to a two-dimensional bin-packing problem \cite{pisinger2007using}, where edge nodes are the bins and the DNN model-variants are the objects to pack. The inference assignment problem can be obtained by relaxing the model $P$: we keep the constraint (\ref{const:loading}), remove the variables $n_{mv}^{e}$, and minimize the average latency. The problem is equivalent to the Generalized  Assignment  Problem, which is NP-hard \cite{yagiurageneralized}.
\end{proof}

\begin{table}[htbp]\caption{Notations}
\begin{center}
\begin{tabular}{r c p{6cm} }
\toprule
\multicolumn{3}{c}{\textbf{Sets}}\\
$\mathcal{N_I}$ & & Set of IoT nodes\\
$\mathcal{N_E}$ & & Set of Edge nodes\\
$\mathcal{M}$ & &Set of models\\
$\mathcal{V}_{m}$ &  & Set of variants of model $m,m\in M$.\\

\multicolumn{3}{c}{\textbf{Parameters}}\\
 
$IL_{mv}^{e}$& &Inference Latency of a request  on variant $v$ of model $m$ in node $e$\\
$P_{ie}$ & & The set of links in the shortest path between the node $i$ and $e$  \\
$CL^{ie}$ & & Communication Latency from node $i$ to node $e$ \\
$\alpha_{mv}$ & & Interference Coefficient of variant $v$ of model $m$ co-located with other model-variants.  \\
$r_{m}^{i}$ &  & request rate from node $i$  on model $m$ \\
$Load_{mv}$ &  & Maximum load on model variant $mv$ \\
$Lmax_{m}^{i}$ &  & Inference Latency requirement of requests on model $m$ from node $i$  \\
$R_{mv}$ & & Memory required for loading the variant $v$ of model $m$\\
$C^{e}$ & & Memory capacity of node $e$ \\
\multicolumn{3}{c}{\textbf{Decision Variables}}\\
$x_{mv}^{ie}$ & $=$ & \(\left\{\begin{array}{rl}
1,  & \text{if variant $v$ of model $m$ is located in node} \\
    & \text{$e$ and serving requests from node $i$}\\
0,  & \text{otherwise} \end{array} \right.\)\\
$n_{mv}^{e}$ & $\in $ & $\mathbb{N}, $ Number of deployed instances of model variant $mv$ in node $e$ \\
$z^{e}$ & $\in $ & $\mathbb{R^{+}}, $ Utilization cost of node $e$\\
$K$ & & Maximum number of model-variants per edge node\\
$\alpha$ & & Weight of the average latency in the objective function\\
\bottomrule
\end{tabular}
\end{center}
\label{tab:TableOfNotationForMyResearch}
\end{table}

\subsection{Discussion on Resource Sharing} \label{sec:DRS}

For inference serving systems which deploy ML models, devices like GPU, TPU, dedicated accelerators are used due to their high performances and as of now most of them work exclusively for one ML model at a time \cite{romero2019infaas}. In the literature, recent works have been proposed. Google Research has adapted DNN inference to run on top of mobile GPU \cite{lee2019device}. Similarly, Amazon Web Services proposed a solution to run inference models on integrated GPUs at the edge \cite{wang2019unified}. But for this kind of applications, besides running inference models on GPU accelerators, it is necessary to consider GPU sharing as well, in order to allocate resources efficiently, opening another area of research. This approach is set to improve on the low utilization and scaling performances of unshared access to a GPU. That idea of GPU sharing can be promising as seen in \cite{jain2018dynamic}, where authors studied the performance of temporal and spacial GPU sharing and \cite{gu2019tiresias}, which presented a GPU cluster manager enabling GPU sharing for DL jobs.\\

The resource sharing such as previously mentioned GPU sharing impacts the inference analysis. The resource allocation required to deploy ML algorithms is complex task, especially in edge computing.  To this end, the emerging new container-based lightweight virtualization technologies allow for separating the model instances that would run in parallel in the same machine. In general this means that resource management systems can scale-up and scale-down allocated resources based on the load variation using these new virtualization technologies. How to effectively share resources across various ML models is an open issues, not only in the context of scalability but also due to the additional latency in ML Inference. As an example, studying the impact of GPU sharing on the performance of ML models is highly important especially on how to scale-up and scale-down resources and how to choose the best model-variant. Due to the complexity of the GPU analysis, which requires a detailed study of numerous existing benchmarks with different ML models, different batch sizes, and GPU memory limitations for our application interference model, in this paper we only use a simplified analysis of the effects that replication and co-locating of model-variant can have on the inference latency.

To this end, we propose one scenario for calculating the inference latency due to resource sharing. We assume parallel usage of hardware in terms of resource sharing. This approach still remains unexplored in edge computing. For model simplicity, we consider that the inference latency of a model-variant increases linearly in terms of the latency of co-located model-variants.  

\begin{equation}
L_{mv1}^{mv2,...,mvK}=L_{mv1} + \sum_{k=2}^{K}\alpha_{k}L_{mvk}
\end{equation}

where $L_{mv1}^{mv2,...,mvK}$ is the inference latency of model-variant $mv1$ in presence of $mv2$,...,$mvK$ and $L_{mv1}, L_{mv2}$ are the inference latency of $mv1$ and $mvk$ running exclusively in the device, respectively. $\alpha_{k}$ is a coefficient called \textit{interference coefficient}  of model variant $mv1$ in presence of model variant $mvk$. This coefficient is introduced to estimate the latency of co-located models in terms of the latency of models running exclusively in the hardware. 

\section{Numerical Analysis} \label{sec:results}
In this section, we evaluate our optimization model using two problem instances $P_1$ and $P_2$, based on  MANIAC  mobile  ad hoc network. Table \ref{tab:t} shows the topology of each studied problem. For each problem, we choose the DNN models randomly from a pre-defined list. The communication latency is obtained from \cite{qin2018sdn}, which was estimated to have a random value with an average  $12.23\; ms$ per link. The inference latency of each model-variant was measured on a GPU GTX 1050 Ti using tensorflow framework. As described in subsection \ref{sec:DRS} we consider the case in which the inference latency increases linearly in terms of the latency of co-located model-variants. 

\begin{table}\caption{Topologies}
\vspace{-0.4 cm}
\begin{center}\label{tab:t}
\begin{tabular}{lllll}
\hline
Problem    & $N_I$ & $N_E$ & $M$ & $V_m$  \\
\hline
P1         &  10   &   5   & 3  & 8   \\
P2         &  11   &   4   & 4  & 8   \\
\hline
\end{tabular}
\end{center}
\end{table}
\vspace{-0.2 cm}
We set weight $\alpha$ equal to $0.1$, node memory capacity $C^e=0.1,\;\; \forall e\in \mathcal{N_E}$ and interference coefficient $\alpha_{mv}=0.1,\;\; \forall m\in \mathcal{M}, \; \forall v\in \mathcal{V}_{m}$. We test our optimization model using different request rates, which correspond to the average of requests per node (e.g $E(r)=5.5$ where $r$ denotes the random variable of request rates). Table \ref{tab:t2} shows the optimization results of latency and cost for two problems with assumed GPU sharing and $K$=2 co-located model-variants. By modifying parameters reported in this table we want to observe the response of the analysed network in terms of latency, cost and utilization. This includes latency measurements for different number of co-locations, as well as measuring the inference load impact on latency, cost and utlization, reported in following subsections. 

\begin{table}\caption{Average latency and cost for $\mathcal{K}$=2  co-located model-variants}
\vspace{-0.3 cm}
\begin{center}\label{tab:t2}
\begin{tabular}{llllllll}
\hline
\multirow{2}{*} &
      \multicolumn{2}{c}{$E(r)=5.5$} &
      \multicolumn{2}{c}{$E(r)=22$} &
      \multicolumn{2}{c}{$E(r)=33$} \\
      
Latency & C &  L  &  C &  L &  C &  L  \\
\hline
$P1$ &
\multicolumn{6}{c}{($\alpha=0.1$, $\alpha_{mv}=0.1$, $C^e=8$, $K=2$)} \\
\hline
  &    0.055 &  17.66  & 0.083 & 18.65 & 0.112 & 37.10\\
\hline
$P2$ &
\multicolumn{6}{c}{($\alpha=0.1$, $\alpha_{mv}=0.1$, $C^e=8$, $K=2$)} \\
\hline
  & 0.051  &  13.38 &  0.051  & 13.43 & 0.06 &14.10 \\
\hline
\end{tabular}
\end{center}
\end{table}
\vspace{-0.2 cm}
\subsection{Impact of the number of co-location}

In our model, we set the maximum number of co-located DNN model-variants $K$ ($K=1$ means that no GPU sharing is allowed). We set the configuration parameters similar to the previous experiment of $P_1$, and we vary the number of co-location $K$ from 1 to 4. Figure \ref{fig:nb-colocation} shows that the average latency decreases with an increase of $K$ until it reaches a maximum value. For our use case, the studied network is small, which allows the optimization to converge when $K$ is equal to 4.  Increasing the $K$ value in this network does not further improve the results, but would be interesting to test in a larger network. For low load increasing model co-location decreases the average latency by 33\% of millisecond-scale per request, and for high load, by  21\%. This result proves that GPU sharing can improve the average latency of inference requests. Decreasing the inference latency by optimally managing the DNN model placement is an interesting result because it allows the system to satisfy latency-critical applications like augmented reality and online games. 

In this paper we used a simplified analysis of the effects that replication and co-locating of model-variant can have on the inference latency. This will be extended in future work to include different scenarios of how the inference latency could potentially behave as a result of resource sharing.
\begin{figure}[t]
 \centering
    \includegraphics[width=0.4\textwidth]{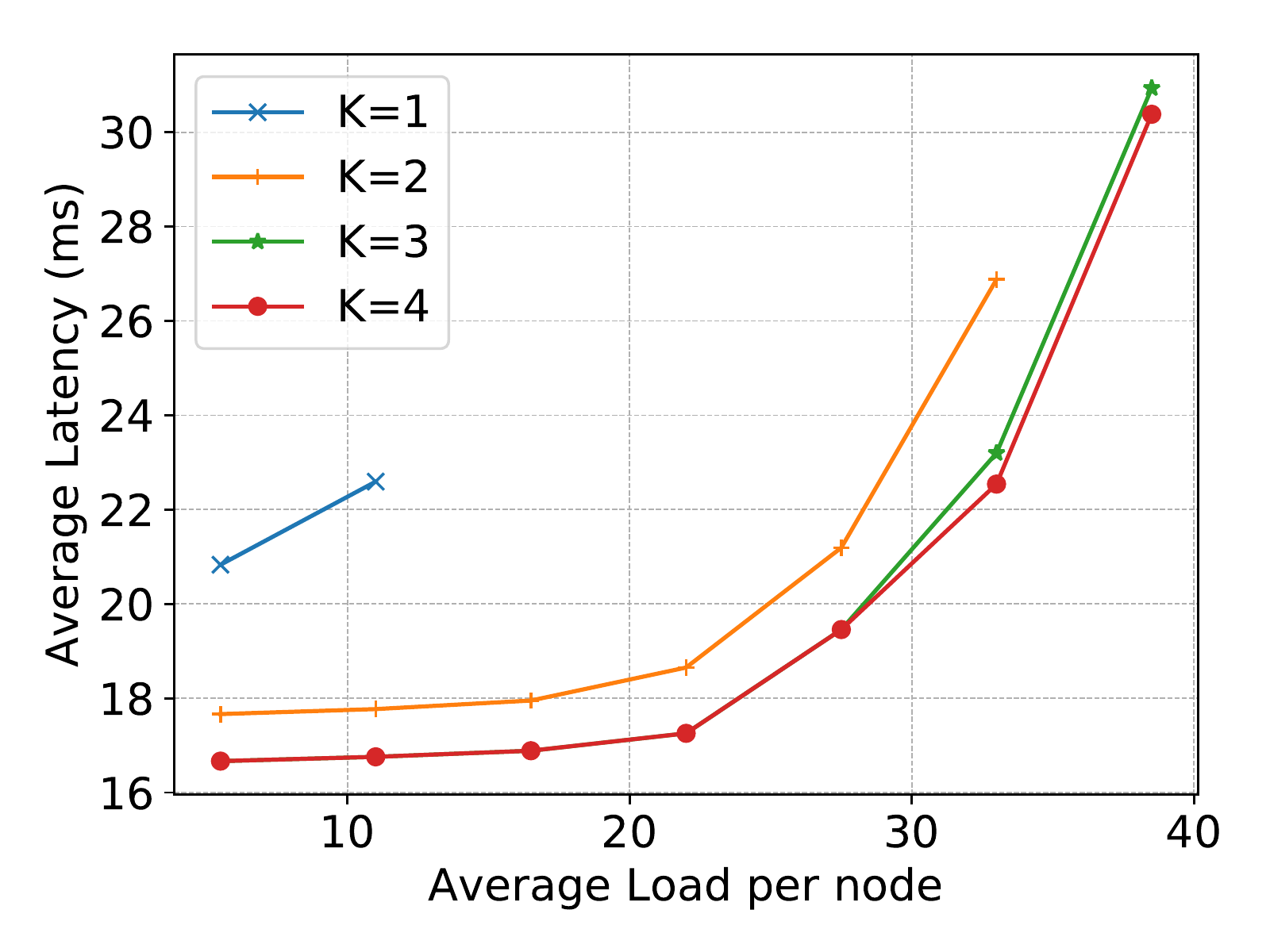}
  \vspace{-0.3 cm}
 \caption{Impact of the number of co-location on latency}
\label{fig:nb-colocation}
\vspace{-0.3 cm}
\end{figure}
\subsection{Impact of Inference Load}
We set the configuration parameters similar to the experiment of
 $P_1$ in \ref{sec:results}, $C^e=5$, and we vary the average load per node. 
 Figure \ref{fig:load} shows that the average latency  varies slightly for low loads while the cost is linearly increasing: when the average load per node increases from 5.5 to 22 (300\% increase), the average latency increases from 18 ms to 19 ms (5.5\% increase), however the cost increases from 0.05 to 0.20 (300\% increase), and the utilization from 38\% to 63\% (65\% increase). This result means that the optimization tends to keep the allocation decision of DNN models while upgrading its variant type to bigger ones, which have higher throughput and higher memory size. Then, when the load is high, the average latency starts to sharply increase: when the average load per node increases from 22 to 33 (50\% increase), the average latency increases from 19 ms to 27 ms (42\% increase), the cost increases from 0.2 to 0.4 (100\% increase), and the utilization from 63\% to 76\% (20\% increase). These results mean that the optimization tends to satisfy inference requests by allocating new model-variants in distant edge nodes that have enough capacity to host the instances. 
\begin{figure}[t]
 \centering
    \includegraphics[width=0.4\textwidth]{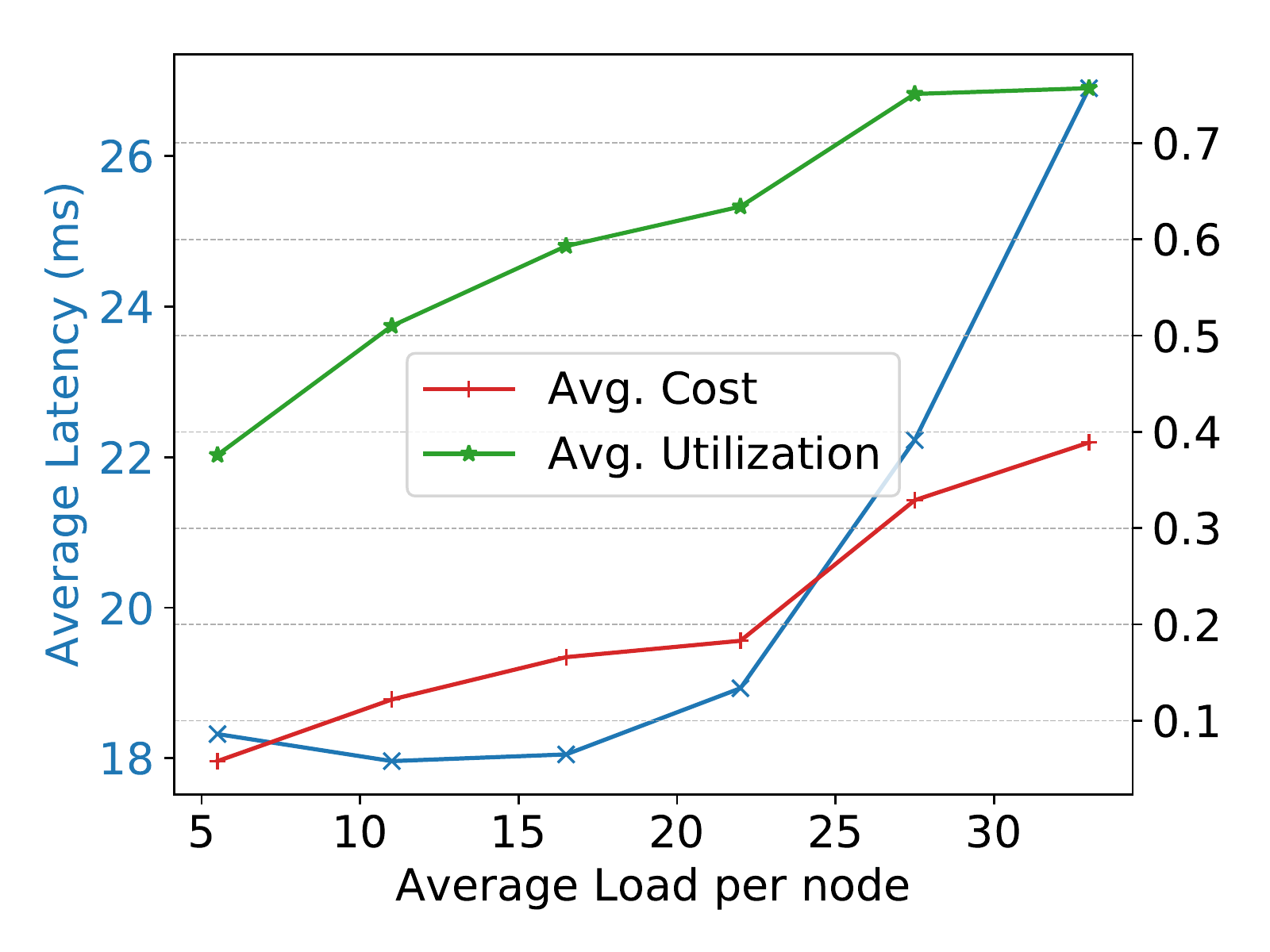}
  \vspace{-0.3 cm}
 \caption{Impact of inference request load on latency and cost in MANIAC network.}
\label{fig:load}
\vspace{-0.2 cm}
\end{figure}

\subsection{Trade-off between the Average Latency and  Cost Function}
We set the configuration parameters similar to the previous experiment, setting the average request load to 27.5, and we alter the weight of the latency and cost in the 
 objective function to evaluate the trade-off between the two objectives. 
 Figure \ref{fig:alpha} shows the opposite behavior of the average latency 
 and the average utilization cost. When we consider only the cost ($\alpha=0$), 
 the optimization tends to allocate the least possible number of instances of each model that can satisfy inference requests. This results in a high value of latency due to the assignment of IoT nodes to distant edge nodes: average latency is 31 ms, average utilization is 67\% and average cost is 0.23. When we increase 
 the value of $\alpha$, we consider more the latency in the decision making. An increase in $\alpha$ cause an increase in the value of the cost and a reduction in the average latency, until a maximum value in which the two objectives converge ($\alpha = 0.2$). The intersection of Pareto optimal curves, i.e. curves of the two goals: average latency and average cost, happens when $\alpha$ is equal to 0.04 with latency equal to 22.7 ms, cost equal to 0.3 and utilization equal to 70\%. Setting the configuration at the intersection point of goals, decreases the cost by 40\% compared to a configuration that set a slightly higher value of $\alpha$. It is worth it to mention that we show the average utilization curve in the figure because it represents a significant metric while it cannot replace the average cost as the intersection between goals is different than the intersection between latency and utilization. 
 The optimization tends to allocate multiple instances from the same model on different edge nodes (possibly with different variant types in each edge node depending on the possible capacity).

\begin{figure}[t]
 \centering
    \includegraphics[width=0.4\textwidth]{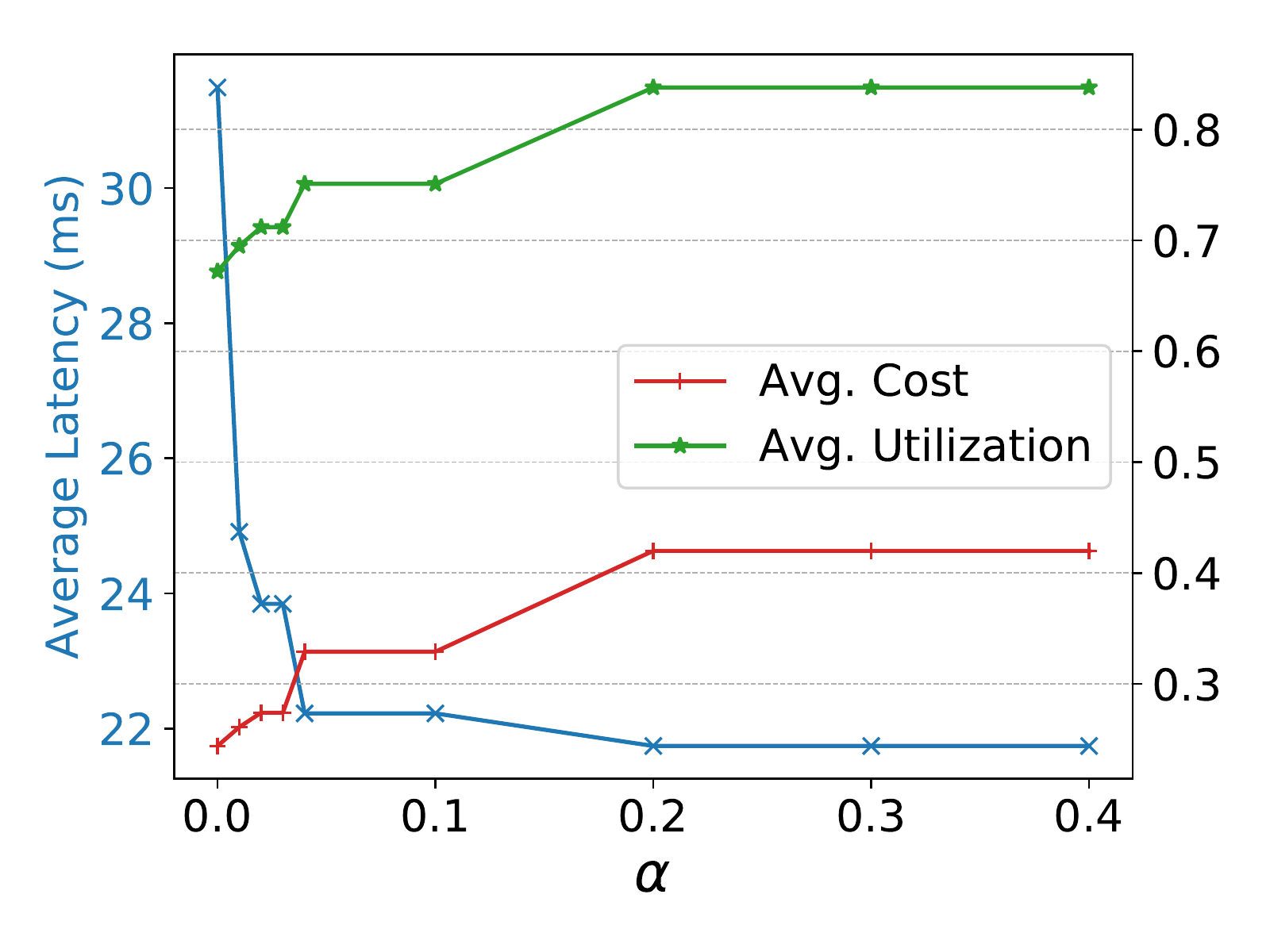}
  \vspace{-0.2 cm}
 \caption{Trade-off between the latency and  cost in MANIAC network.}
\label{fig:alpha}
\vspace{-0.2 cm}
\end{figure}

\section{Conclusion} \label{sec:conclusion}

In this paper, we studied the DNN Model Variant Selection and Placement (MVSP) in edge computing networks. A mathematical model was proposed to formulate the problem  considering the inference latency of different model-variants, the communication latency between nodes, and the utilization cost of edge computing nodes (resources).
Also, we considered the effects of hardware sharing on inference latency regarding GPU edge computing nodes shared between different model-variants. We studied the placement  results of the optimization and its effect on the average latency and cost. We showed that GPU sharing is a valuable approach to handle the increase of inference request rate. Results show that: for low load increasing model co-location decreases the average latency by 33\% of millisecond-scale per request, and for high load, by  21\%. We plan to further work on extending our model to consider more GPU sharing scenarios, and analyze parameters like multiple frameworks, and multiple hardware devices, as well as implementing heuristic solutions.

\bibliographystyle{IEEEtran}
\bibliography{SecurityBib}

\end{document}